\documentclass[american,reprint,aps,pra,superscriptaddress,showpacs]{revtex4-1}
\usepackage[latin9]{inputenc}
\usepackage{babel}
\usepackage{amsthm}
\usepackage{amsmath}
\usepackage{amssymb}
\usepackage{graphicx}
\usepackage[unicode=true,pdfusetitle,
 bookmarks=true,bookmarksnumbered=false,bookmarksopen=false,
 breaklinks=false,pdfborder={0 0 0},backref=false,colorlinks=false]
 {hyperref}
\hypersetup{
 colorlinks,linkcolor=myurlcolor,citecolor=myurlcolor,urlcolor=myurlcolor}

\makeatletter
\@ifundefined{textcolor}{}
{%
 \definecolor{BLACK}{gray}{0}
 \definecolor{WHITE}{gray}{1}
 \definecolor{RED}{rgb}{1,0,0}
 \definecolor{GREEN}{rgb}{0,1,0}
 \definecolor{BLUE}{rgb}{0,0,1}
 \definecolor{CYAN}{cmyk}{1,0,0,0}
 \definecolor{MAGENTA}{cmyk}{0,1,0,0}
 \definecolor{YELLOW}{cmyk}{0,0,1,0}
 }
\theoremstyle{plain}
\newtheorem{thm}{\protect\theoremname}

\usepackage{babel}
\usepackage{times}\usepackage{txfonts}\usepackage{braket}\usepackage{colortbl}\definecolor{myurlcolor}{rgb}{0,0,0.7}

\providecommand{\theoremname}{Theorem}

\makeatother

\providecommand{\theoremname}{Theorem}

\begin{document}

\title{Measuring Quantum Coherence with Entanglement}

\author{Alexander Streltsov}

\email{alexander.streltsov@icfo.es}

\affiliation{ICFO-The Institute of Photonic Sciences, Mediterranean Technology
Park, 08860 Castelldefels (Barcelona), Spain}

\author{Uttam Singh}
\email{uttamsingh@hri.res.in}

\affiliation{Harish-Chandra Research Institute, Chhatnag Road, Jhunsi, Allahabad
211 019, India}

\author{Himadri Shekhar Dhar}
\email{himadrisdhar@hri.res.in}

\affiliation{Harish-Chandra Research Institute, Chhatnag Road, Jhunsi, Allahabad
211 019, India}

\affiliation{School of Physical Sciences, Jawaharlal Nehru University, New Delhi
110 067, India}

\author{Manabendra Nath Bera}
\email{manabbera@hri.res.in}

\affiliation{Harish-Chandra Research Institute, Chhatnag Road, Jhunsi, Allahabad
211 019, India}

\author{Gerardo Adesso}

\email{gerardo.adesso@nottingham.ac.uk}

\affiliation{$\mbox{School of Mathematical Sciences, The University of Nottingham, University Park, Nottingham NG7 2RD, United Kingdom}$}
\begin{abstract}
Quantum coherence is an essential ingredient in quantum information
processing and plays a central role in emergent fields such as nanoscale
thermodynamics and quantum biology. However, our understanding and
quantitative characterization of coherence as an operational resource
are still very limited. Here we show that any degree of coherence with respect to some reference basis can be converted to entanglement via incoherent operations. This finding allows us
to define a novel general class of measures of coherence for a quantum
system of arbitrary dimension, in terms of the maximum bipartite entanglement
that can be generated via incoherent operations applied to the system
and an incoherent ancilla. The resulting measures are proven to be
valid coherence monotones satisfying all the requirements dictated
by the resource theory of quantum coherence. We demonstrate the usefulness of our approach by proving that the fidelity-based geometric measure of coherence is a full convex coherence monotone, and deriving a closed formula for it on arbitrary single-qubit states. Our work provides a clear
quantitative and operational connection between coherence and entanglement,
two landmark manifestations of quantum theory and both key enablers
for quantum technologies.
\end{abstract}


\date{June 10, 2015}

\maketitle
\noindent \textbf{\emph{Introduction}}\textbf{.}---Coherence is a
fundamental aspect of quantum physics that encapsulates the defining
features of the theory \cite{Leggett1980}, from the superposition
principle to quantum correlations. It is a key component in various
quantum information and estimation protocols and is primarily accountable
for the advantage offered by quantum tasks versus classical ones \cite{Giovannetti2011,Nielsen10}.
In general, coherence is an important physical resource in low-temperature
thermodynamics \cite{Aberg14,Varun14,Oppenheim14,Rudolf114,Rudolf214},
for exciton and electron transport in biomolecular networks \cite{Plenio2008,Aspuru2009,Lloyd2011,Li2012,Huelga13,Levi14},
and for applications in nanoscale physics \cite{VazquezH2012,Karlstrom11}.
Experimental detection of coherence in living complexes \cite{Engel2007,Collini2010}
and creation of coherence in hot systems \cite{Rybak2011} have also
been reported.

While the theory of quantum coherence is historically well developed
in quantum optics \cite{Glauber63,Sudarshan63,Mandel95,Kim02,Asboth05,Vogel1,Vogel2,Vogel3}, a rigorous framework to quantify coherence for general
states in information theoretic terms has only
been attempted recently \cite{Baumgratz2014,Levi14,Girolami14,Vogel2,Smyth14,Diego15}.
This framework is based on identifying the set of incoherent
states and a class of `free' operations, named incoherent quantum
channels, that map the set onto itself \cite{Baumgratz2014,Levi14}.
The resulting resource theory of coherence is in direct analogy
with the resource theory of entanglement \cite{HorodeckiRMP09},
in which local operations and classical communication are the `free' operations that map the set of separable states onto itself \cite{Fernando2013}. Within such a framework for coherence,
one can define suitable measures that vanish for any incoherent state,
and satisfy specific monotonicity requirements under incoherent
channels. Measures that respect these conditions gain the attribute
of coherence monotones, in analogy with entanglement monotones \cite{Vidal2000}.
Examples of coherence monotones include the relative entropy and the
$l_{1}$-norm of coherence \cite{Baumgratz2014}.
Intuitively, both coherence and
entanglement  capture quantumness of a physical system, and it
is well known that entanglement stems from the superposition principle,
which is also the essence of coherence. It is then legitimate to ask
how can one resource emerge \textit{quantitatively} from the other
\cite{Asboth05,Vogel2}.

\begin{figure}[t]
\centering \includegraphics[width=8.5cm]{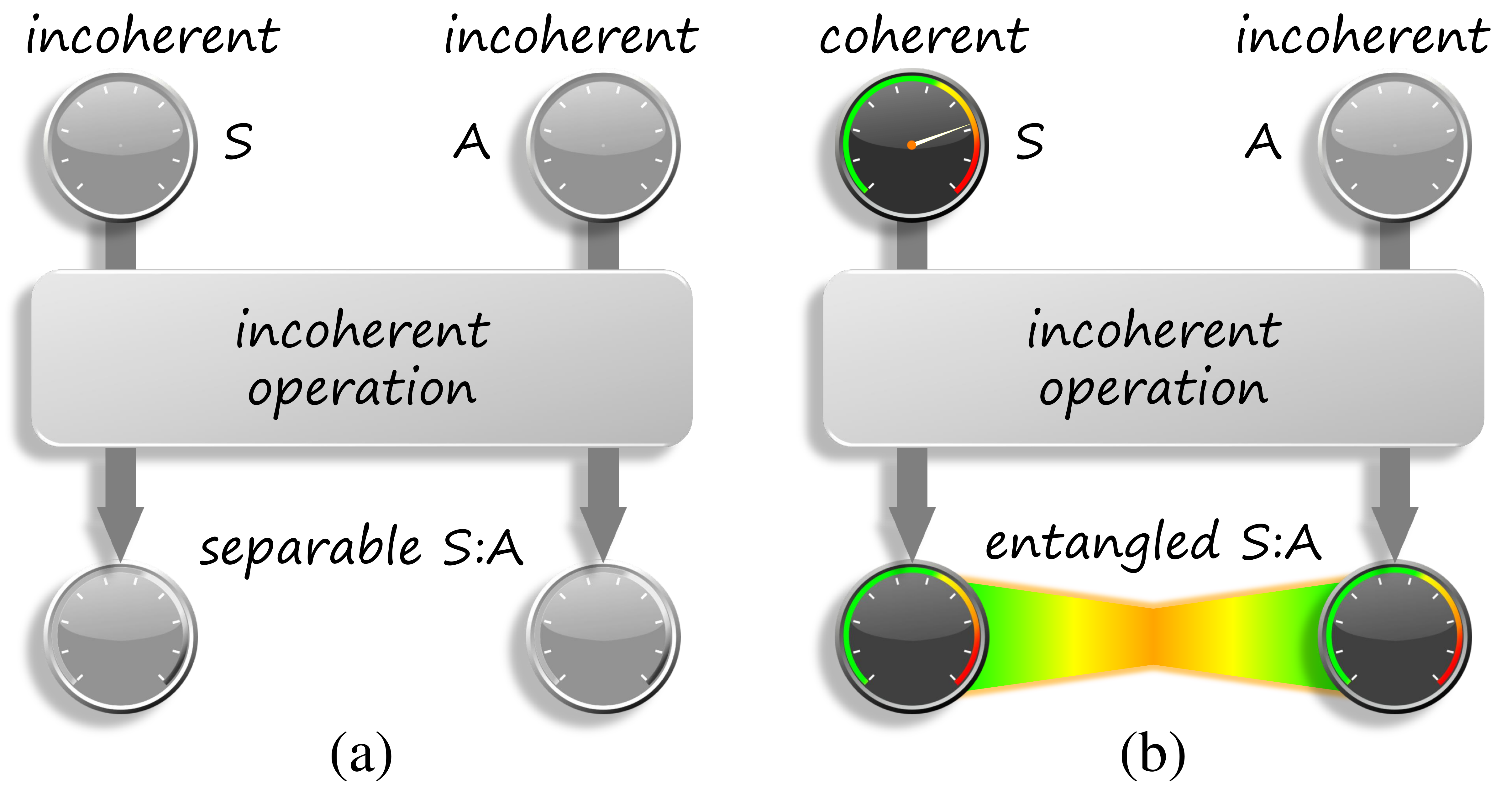} \caption{(a)
Incoherent operations cannot generate entanglement from incoherent input
states. (b) Conversely, we show that any nonzero coherence in the
input state of a system $S$ can be converted to entanglement
via incoherent operations on $S$ and an incoherent ancilla $A$.
Input coherence and output entanglement are quantitatively equivalent:
For any entanglement monotone $E$, the maximum entanglement generated between $S$ and $A$ by incoherent operations defines
a faithful coherence monotone $C_{E}$ on the initial state of $S$.}

\label{cofig}
\end{figure}


In this Letter, we provide a mathematically rigorous approach to resolve
the above question, using a common frame
to quantify quantumness in terms of coherence and entanglement. In
particular,  we show that any nonzero amount
of coherence in a system $S$ can be converted to (distillable)
entanglement between $S$ and an initially incoherent ancilla $A$,
by means of incoherent operations. This allows us to formulate a novel, general
method to quantify coherence in terms of entanglement (see Fig.~\ref{cofig}). Namely we prove that, given an arbitrary set of entanglement
monotones $\{E\}$, one can define a corresponding class of coherence
monotones $\{C_{E}\}$ that satisfy all the requirements of Ref.~\cite{Baumgratz2014}.
The input coherence $C_{E}$ of $S$ is  defined as the
maximum output entanglement $E$ over all incoherent operations on
$S$ and $A$. We explicitly evaluate the maximization in some relevant instances, defining novel coherence monotones such as the fidelity-based geometric measure of coherence. These results provide powerful advances for the operational quantification of  coherence.

\smallskip{}

\noindent\textbf{\emph{Characterizing coherence}}\textbf{.}---For
an arbitrary fixed reference basis $\{\ket{i}\}$, the incoherent
states are defined as \cite{Baumgratz2014}
\begin{equation}
\sigma={\sum}_{i}p_{i}\ket{i}\bra{i},
\end{equation}
where $p_{i}$ are  probabilities. Any state which cannot
be written as above is defined {\it coherent} \cite{Baumgratz2014}.
Note that, unlike other resources in information theory, coherence is basis-dependent. The reference basis with respect to which coherence is measured depends on the physical problem under investigation; it is e.g.~the energy basis for transport phenomena in engineered and biological domains \cite{Huelga13}, or the eigenbasis of the generator of an unknown phase shift in quantum metrology \cite{Giovannetti2011}.

A completely positive trace preserving map $\Lambda$ is said to be
an incoherent operation if it can be written as
\begin{equation}
\Lambda[\rho]={\sum}_{l}K_{l}\rho K_{l}^{\dagger},\label{eq:incoherent_operation}
\end{equation}
where the defining operators $K_{l}$, called incoherent Kraus operators, map every incoherent
state to some other incoherent state, i.e.~$K_{l}{\cal \mathcal{I}}K_{l}^{\dagger}\subseteq\mathcal{I}$, where ${\cal {I}}$ is the set
of incoherent states.

Following established notions from entanglement theory \cite{Vedral1997,Vedral1998,Plenio2007,HorodeckiRMP09},
Baumgratz \emph{et al}. proposed the following postulates for a measure
of coherence $C(\rho)$ in Ref.~\cite{Baumgratz2014}:
\begin{itemize}
\item {(C1)} $C(\rho)\geq0$, and $C(\sigma)=0$ if and only if $\sigma\in\mathcal{I}$.\\[-18pt]
\item {(C2)} $C(\rho)$ is nonincreasing under incoherent operations,
i.e., $C(\rho)\geq C(\Lambda[\rho])$ with $\Lambda[{\cal {I}]\subseteq{\cal {I}}}$.\\[-18pt]
\item {(C3)} $C(\rho)$ is nonincreasing on average under selective incoherent
operations, i.e., $C(\rho)\geq\sum_{l}p_{l}C(\varsigma_{l})$, with probabilities
$p_{l}=\mathrm{Tr}[K_{l}\rho K_{l}^{\dagger}]$, states $\varsigma_{l}=K_{l}\rho K_{l}^{\dagger}/p_{l}$,
and incoherent Kraus operators $K_{l}$ obeying $K_{l}\mathcal{I}K_{l}^{\dagger}\subseteq\mathcal{I}$.\\[-18pt]
\item {(C4)} $C(\rho)$ is a convex function of density matrices, i.e.,
$C(\sum_{i}p_{i}\rho_{i})\leq\sum_{i}p_{i}C(\rho_{i})$.
\end{itemize}
Note that conditions (C3) and (C4) automatically
imply condition (C2). The reason we listed all conditions above
is that (similar to entanglement measures) there might exist meaningful
quantifiers of coherence which satisfy conditions (C1) and (C2), but
for which conditions (C3) and (C4) are either violated or cannot be
proven. Following the analogous notion from entanglement theory, we
call a quantity which satisfies conditions (C1), (C2), and (C3) a
\emph{coherence monotone}.

Examples of functionals that satisfy all the four properties mentioned
above include the $l_{1}$-norm of coherence \cite{Baumgratz2014}
$C_{l_{1}}(\rho)=\sum_{i\neq j}|\rho_{ij}|$ and the relative entropy
of coherence \cite{Baumgratz2014}
\begin{equation}
C_{\mathrm{r}}(\rho)=\min_{\sigma\in{\cal I}}H(\rho||\sigma)\label{eq:Cr}
\end{equation}
with the quantum relative entropy $H(\rho||\varsigma)=\mathrm{Tr}[\rho\log_{2}\rho]-\mathrm{Tr}[\rho\log_{2}\varsigma]$.
As was shown in \cite{Baumgratz2014}, the relative entropy of coherence
can also be written as $C_{\mathrm{r}}(\rho)=H(\rho_{d})-H(\rho)$,
where $\rho_{d}$ is the diagonal part of the density matrix $\rho$
in the reference basis $\{\ket{i}\}$ and $H$ is the von Neumann
entropy.

\smallskip{}

\noindent\textbf{\emph{Bipartite coherence}}\textbf{.}---We first
extend the framework of coherence to the bipartite scenario (see also
\cite{Frozen2014}); the following definitions extend straightforwardly to multipartite systems. In particular, for a bipartite system with two
subsystems $X$ and $Y$, and with respect to a fixed reference product
basis $\{\ket{i}^{X}\otimes\ket{j}^{Y}\}$, we define bipartite incoherent
states as follows:
\begin{equation}
\rho^{XY}={\sum}_{k}p_{k}\sigma_{k}^{X}\otimes\tau_{k}^{Y}.\label{eq:incoherent_state}
\end{equation}
Here, $p_{k}$ are probabilities and the states $\sigma_{k}^{X}$
and $\tau_{k}^{Y}$ are incoherent states on the subsystem $X$ and
$Y$ respectively, i.e.~$\sigma_{k}^{X}=\sum_{i}p'_{ik}\ket{i}\bra{i}^{X}$
and $\tau_{k}^{Y}=\sum_{j}p''_{jk}\ket{j}\bra{j}^{Y}$ for probabilities
$p'_{ik}$ and $p''_{jk}$. Note that the states in Eq.~(\ref{eq:incoherent_state}) are always separable.

We next define bipartite incoherent operations as
in Eq.~(\ref{eq:incoherent_operation}), with incoherent Kraus operators
$K_{l}$ such that $K_{l}{\cal I}K_{l}^{\dagger}\subseteq{\cal I}$,
where ${\cal I}$ is now the set of bipartite incoherent states defined
in Eq.~(\ref{eq:incoherent_state}).
An example of bipartite incoherent operation is the two-qubit
CNOT gate $U_{\mathrm{CNOT}}$. It is not possible to create coherence
from an incoherent two-qubit state by using the CNOT gate, since it
takes any pure incoherent state $\ket{i}\otimes\ket{j}$ to another
pure incoherent state, $
U_{\mathrm{CNOT}}\left(\ket{i}\otimes\ket{j}\right)=\ket{i}\otimes\ket{\mathrm{mod}\left(i+j,2\right)}$.
The CNOT
gate can be used however to create entanglement, e.g.~it is well known that the state $U_{\mathrm{CNOT}}(\ket{\psi}\otimes\ket{0})$ is entangled
for any coherent state $\ket{\psi}$ \cite{Nielsen10}.

\smallskip{}

\noindent\textbf{\emph{Converting coherence to entanglement}}\textbf{.}---Referring
to Fig.~\ref{cofig}, we say that the coherence in the initial state $\rho^{S}$ of a (finite-dimensional)
system $S$ can be converted to entanglement via incoherent operations if, by attaching
an ancilla $A$ initialized in a reference incoherent state $\ket{0}\bra{0}^{A}$,
the final system-ancilla state $\Lambda^{SA}[\rho^{S}\otimes\ket{0}\bra{0}^{A}]$
is entangled for some incoherent operation $\Lambda^{SA}$. Note that
incoherent system states $\sigma^{S}$ cannot be used for conversion to entangled states in this way, since for any incoherent state $\sigma^{S}$ the
state $\Lambda^{SA}[\sigma^{S}\otimes\ket{0}\bra{0}^{A}]$ will be of
the form given in Eq.~(\ref{eq:incoherent_state}), and thus separable.

Entanglement can instead be generated by incoherent operations if the initial  $\rho^{S}$ is coherent, as in the two-qubit CNOT example. It is then natural to ask: Can {\it any} nonzero amount of coherence be converted to entanglement via incoherent operations?  To answer this, we first consider distance-based
measures of entanglement $E_{D}$ and coherence $C_{D}$
\cite{Vedral1997,Vedral1998,Plenio2007,Baumgratz2014,Frozen2014}:
\begin{equation}\label{ECD}
E_{D}(\rho)=\min_{\chi\in{\cal S}}D(\rho,\chi),\qquad C_{D}(\rho)=\min_{\sigma\in{\cal I}}D(\rho,\sigma).
\end{equation}
Here, ${\cal S}$ is the set of separable states and ${\cal I}$ is
the set of incoherent states. Moreover, we demand that the distance
$D$ be contractive under quantum operations,
\begin{equation}
D(\Lambda[\rho],\Lambda[\varsigma])\leq D(\rho,\varsigma)\label{eq:contractive}
\end{equation}
for any completely positive trace preserving map $\Lambda$. This
implies that $E_{D}$ does not increase under local operations and
classical communication \cite{Vedral1997,Vedral1998}, and $C_{D}$
does not increase under incoherent operations \cite{Baumgratz2014}.
Equipped with these tools we are now in position to present the first
result of this Letter.
\begin{thm}
\label{thm:1}
\vspace*{-2pt}
For any contractive distance $D$, the amount of (distance-based)
entanglement $E_{D}$ generated from a state $\rho^{S}$ via an incoherent
operation $\Lambda^{SA}$ is bounded above by its (distance-based)
coherence $C_{D}$:
\begin{equation}
\mbox{\ensuremath{E_{D}^{S:A}\left(\Lambda^{SA}\left[\rho^{S}\otimes\ket{0}\bra{0}^{A}\right]\right)\leq C_{D}\left(\rho^{S}\right)}}.\label{eq:bound}
\end{equation}
\end{thm}
\vspace*{-12pt}
\begin{proof} Let $\sigma^{S}$ be the closest incoherent state to
$\rho^{S}$, i.e., $C_{D}(\rho^{S})=D(\rho^{S},\sigma^{S})$. The
contractivity of the distance $D$ further implies the equality: $D(\rho^{S},\sigma^{S})=D(\rho^{S}\otimes\ket{0}\bra{0}^{A},\sigma^{S}\otimes\ket{0}\bra{0}^{A})$.
In the final step, note that the application of an incoherent operation
$\Lambda^{SA}$ to the incoherent state $\sigma^{S}\otimes\ket{0}\bra{0}^{A}$
brings it to another incoherent---and thus separable---state. Applying
Eq.~(\ref{eq:contractive}) and combining the aforementioned results
we arrive at the desired inequality: 
$C_{D}(\rho^{S})\geq D(\Lambda^{SA}[\rho^{S}\otimes\ket{0}\bra{0}^{A}],\Lambda^{SA}[\sigma^{S}\otimes\ket{0}\bra{0}^{A}])\geq E_{D}^{S:A}(\Lambda^{SA}[\rho^{S}\otimes\ket{0}\bra{0}^{A}])$.
 \end{proof} \vspace*{-2pt}
This result
provides a strong link between the resource frameworks of entanglement and coherence. An even stronger link exists when choosing specifically $D$ as the relative entropy.
The corresponding quantifiers are the relative entropy of entanglement
$E_{\mathrm{r}}$ \cite{Vedral1997}, and the relative entropy of
coherence $C_{\mathrm{r}}$ \cite{Baumgratz2014} introduced
in Eq.~(\ref{eq:Cr}). Importantly, the inequality (\ref{eq:bound})
can be saturated for these measures if the dimension of the ancilla
is not smaller than that of the system, $d_{A}\geq d_{S}$.
In this case there always exists an incoherent operation $\Lambda^{SA}$
such that
\begin{equation}
\mbox{\ensuremath{E_{\mathrm{r}}^{S:A}\left(\Lambda^{SA}\left[\rho^{S}\otimes\ket{0}\bra{0}^{A}\right]\right)=C_{\mathrm{r}}\left(\rho^{S}\right)}}.\label{eq:optimal}
\end{equation}
To prove this statement, we consider the unitary operation
\begin{align}
U & ={\sum}_{i=0}^{d_{S}-1}{\sum}_{j=0}^{d_{S}-1}\ket{i}\bra{i}^{S}\otimes\ket{\mathrm{mod}(i+j,d_{S})}\bra{j}^{A}\nonumber \\
 & +{\sum}_{i=0}^{d_{S}-1}{\sum}_{j=d_{S}}^{d_{A}-1}\ket{i}\bra{i}^{S}\otimes\ket{j}\bra{j}^{A}.\label{eq:CNOT}
\end{align}
Note that for two qubits this unitary is equivalent to the CNOT gate
with $S$ as the control qubit and $A$ as the target qubit. It can
be seen by inspection that this unitary is incoherent (i.e., the state
$\Lambda^{SA}[\rho^{SA}]=U\rho^{SA}U^{\dagger}$ is incoherent for
any incoherent state $\rho^{SA}$), and maps the state $\rho^{S}\otimes\ket{0}\bra{0}^{A}$
to the state
\begin{equation}
\mbox{\ensuremath{\Lambda^{SA}\left[\rho^{S}\otimes\ket{0}\bra{0}^{A}\right]={\sum}_{i,j}\ \rho_{ij}\ket{i}\bra{j}^{S}\otimes\ket{i}\bra{j}^{A}}},
\end{equation}
where $\rho_{ij}$ are the matrix elements of $\rho^{S}=\sum_{i,j}\rho_{ij}\ket{i}\bra{j}^{S}$.
In the next step we use the fact that for any quantum state $\varsigma^{SA}$
the relative entropy of entanglement is bounded below as follows \cite{Plenio2000}:
$E_{\mathrm{r}}^{S:A}(\varsigma^{SA})\geq H(\varsigma^{S})-H(\varsigma^{SA})$. Applied
to the state $\Lambda^{SA}[\rho^{S}\otimes\ket{0}\bra{0}^{A}]$, this
inequality reduces to
\begin{equation}
\mbox{\ensuremath{E_{\mathrm{r}}^{S:A}\left(\Lambda^{SA}\left[\rho^{S}\otimes\ket{0}\bra{0}^{A}\right]\right)\geq H\left({\sum}_{i}\ \rho_{ii}\ket{i}\bra{i}^{S}\right)-H(\rho^{S})}}.
\end{equation}
Noting that the right-hand side of this inequality is equal to the
relative entropy of coherence $C_{\mathrm{r}}(\rho^{S})$ \cite{Baumgratz2014},
we obtain $E_{\mathrm{r}}^{S:A}(\Lambda^{SA}[\rho^{S}\otimes\ket{0}\bra{0}^{A}])\geq C_{\mathrm{r}}(\rho^{S})$.
The proof of Eq.~(\ref{eq:optimal}) is complete by combining this
result with Theorem~\ref{thm:1}.

The results presented above also hold for the distillable entanglement
$E_{\mathrm{d}}$. Namely, the relative entropy of coherence
$C_{\mathrm{r}}$ also serves as an upper bound for the conversion to
distillable entanglement via incoherent operations, and the equality in Eq.~(\ref{eq:optimal}) still holds if $E_{\mathrm{r}}$ is replaced by $E_{\mathrm{d}}$, and the  incoherent unitary of Eq.~(\ref{eq:CNOT}) is applied.
This  follows from Theorem~\ref{thm:1}, together with
the fact that distillable entanglement admits the following bounds \cite{Horodecki2000, Devetak2005}:
 $H(\varsigma^{S})-H(\varsigma^{SA}) \leq E_{\mathrm{d}}^{S:A}\leq E_{\mathrm{r}}^{S:A}$.

This shows that the degree of (relative entropy of) coherence in the
initial state of $S$ can be exactly converted to an equal
degree of (distillable or relative entropy of) entanglement between $S$
and the incoherent ancilla $A$ by a suitable incoherent
operation, that is a generalized CNOT gate. We can now settle the general question posed above.
\begin{thm}
\label{thm:2}
\vspace*{-2pt}
A state $\rho^{S}$ can be converted to an entangled state via incoherent operations if and only if $\rho^{S}$ is coherent.\end{thm}
\begin{proof} \vspace*{-2pt} If $\rho^{S}$ is incoherent, it cannot be converted to an entangled state via incoherent operations by Theorem~\ref{thm:1}.
Conversely, if $\rho^{S}$ is coherent, it has nonzero
relative entropy of coherence $C_{\mathrm{r}}(\rho^{S})>0$. By
Eq.~(\ref{eq:optimal}), there exists an incoherent operation $\Lambda^{SA}$
leading to nonzero relative entropy of entanglement $E_{\mathrm{r}}^{S:A}(\Lambda^{SA}[\rho^{S}\otimes\ket{0}\bra{0}^{A}])>0$, concluding the proof.
\end{proof}
\smallskip{}

\noindent\textbf{\emph{Quantifying coherence with entanglement}}\textbf{.}---We are  ready to present the central result of the Letter. Reversing the perspective, Theorem~\ref{thm:1}
can also be seen as providing a lower bound on distance-based
measures of coherence through conversion to entanglement: precisely,
the coherence degree $C_{D}$ of a state $\rho^{S}$ is always
bounded below by the maximal entanglement degree $E_{D}$ generated
from it by incoherent operations.

Going now beyond the specific
setting of distance-based measures, we will show that such a maximization
of the output entanglement, for any  fully general entanglement
monotone, leads to a quantity which yields a valid quantifier
of input coherence in its own right. We specifically define the family of {\em entanglement-based coherence
measures} $\{C_{E}\}$ as follows:
\begin{equation}
C_{E}(\rho^{S})=\lim_{d_{A}\rightarrow\infty}\bigg\{ \sup_{\Lambda^{SA}}E^{S:A}\left(\Lambda^{SA}\left[\rho^{S}\otimes\ket{0}\bra{0}^{A}\right]\right)\bigg\} .\label{eq:Ce}
\end{equation}
Here, $E$ is an arbitrary entanglement measure and the supremum is taken
over all incoherent operations $\Lambda^{SA}$
\footnote{Note that the limit $d_{A}\rightarrow\infty$ in Eq.~(\ref{eq:Ce})
is well defined, since the supremum $\sup_{\Lambda^{SA}}E^{S:A}\left(\Lambda^{SA}\left[\rho^{S}\otimes\ket{0}\bra{0}^{A}\right]\right)$
cannot decrease with increasing dimension $d_{A}$ of the ancilla. %
}.

It is crucial to benchmark the validity of $C_{E}$ for any $E$
as a proper measure of coherence. Remarkably, we find that $C_{E}$
satisfies all the aforementioned conditions (C1)--(C3) given an arbitrary
entanglement monotone $E$, with the addition of (C4) if $E$ is convex
as well. We namely get the following result:
\begin{thm}
\label{thm:monotone}$C_{E}$
is a (convex) coherence monotone for any (convex) entanglement monotone
$E$. \vspace*{-4pt}\end{thm}
\begin{proof} Using the arguments presented above
it is easy to see that $C_{E}$ is nonnegative, and zero if and only
if the state $\rho^{S}$ is incoherent. Moreover, $C_{E}$ does not
increase under incoherent operations $\Lambda^{S}$ performed on the
system $S$. This can be seen directly from the definition of $C_{E}$
in Eq.~(\ref{eq:Ce}), noting that an incoherent operation $\Lambda^{S}$
on the system $S$ is also incoherent with respect to $SA$. The proof
that $C_{E}$ further satisfies condition (C3) is presented in the
Supplemental Material \cite{epaps}. There we also show that $C_{E}$ is convex
for any convex entanglement monotone $E$, i.e. (C4) is fulfilled
as well in this case. \end{proof}

These powerful results complete
the parallel between coherence and entanglement, \textit{de facto}
establishing their full quantitative equivalence within the respective
resource theories. Thanks to Theorem~\ref{thm:monotone}, one can now use the comprehensive knowledge acquired in entanglement theory in the last two decades \cite{Vedral1997,Vidal2000,Wootters97,Plenio2007}, to address the quantification of coherence in a variety of operational settings, and to define and validate physically motivated coherence monotones.
For instance,  $C_{E}$ as defined by Eq.~(\ref{eq:Ce}) amounts to the previously defined relative entropy
of coherence \cite{Baumgratz2014}, if $E$ is  the relative entropy of entanglement or the
distillable entanglement.

Furthermore, we can now  focus on the relevant case of $E$ being the geometric entanglement \cite{Wei2003,Streltsov2010} $E_{\mathrm{g}}$, defined for a bipartite state $\rho$ as $E_{\mathrm{g}}(\rho)=1-\max_{\chi\in{\cal S}}F(\rho,\chi)$, where
the maximum is taken over all separable states $\chi \in {\cal S}$, and
$F(\rho,\varsigma)=\big(\mathrm{Tr}(\sqrt{\rho}\varsigma\sqrt{\rho}\,)^{1/2}\big)^{2}$ is the
Uhlmann fidelity. The geometric entanglement coincides with its expression obtained via  convex roof  \cite{Horodecki2003,Streltsov2010},
$E_{\mathrm{g}}(\rho)=\min\sum_{i}p_{i}E_{\mathrm{g}}(\ket{\psi_{i}})$,
where the minimum is over all decompositions of  $\rho=\sum_{i}p_{i}\ket{\psi_{i}}\bra{\psi_{i}}$.
In the Supplemental Material \cite{epaps}, we show that the {\it  geometric measure of coherence} $C_\mathrm{g}$, associated to $E_{\mathrm{g}}$ via Eq.~(\ref{eq:Ce}), can be evaluated explicitly and amounts to $C_{\mathrm{g}}(\rho)=1-\max_{\sigma\in{\cal I}}F(\rho,\sigma)$,  where the maximum is taken over all incoherent states $\sigma \in {\cal I}$. The incoherent operation which attains the maximum in Eq.~(\ref{eq:Ce}) is  again the generalized CNOT  defined by Eq.~(\ref{eq:CNOT}). Due to Theorem~\ref{thm:monotone}, since the geometric measure of entanglement is a full convex entanglement monotone \cite{Wei2003,Plenio2007}, we have just proven that the geometric measure of coherence $C_\mathrm{g}$ is a full convex coherence monotone obeying (C1)--(C4). This settles an important question left open in previous literature \cite{Baumgratz2014,Frozen2014}. Remarkably, the geometric measure $C_\mathrm{g}$ is also~analytically computable~for an arbitrary state $\rho$ of one qubit \cite{epaps}, as follows:
\begin{equation}\label{CGrho}
\hbox{$C_{\mathrm{g}}(\rho)=\frac12\left(1-\sqrt{1-4|\rho_{01}|^2}\right)$}\,,
\end{equation}
where $\rho_{01}$ is the off-diagonal element of $\rho$ with respect to the reference basis. Notice that $C_{\mathrm{g}}$ in this case is a simple~monoto-nic~function of the $l_{1}$-norm of coherence \cite{Baumgratz2014}, $C_{l_1}(\rho)=2|\rho_{01}|$.

Some of these results extend to any distance-based entanglement measure $E_{g(F)}$ defined via Eq.~(\ref{ECD}) with $D_{g(F)}(\rho,\varsigma)=g\big(F(\rho,\varsigma)\big)$, where $g(F)$ is a  nonincreasing function of the fidelity $F$. These include the Bures measure of entanglement \cite{Vedral1997,Vedral1998}, with  $g(F)=2(1-\sqrt{F})$, and the Groverian measure of entanglement \cite{Biham2002,Shapira2006}, with $g(F)=\sqrt{1-F}$. For any such entanglement $E_{g(F)}$, the  corresponding quantifier of coherence is $C_{g(F)}(\rho) = \min_{\sigma\in{\cal I}}D_{g(F)}(\rho,\sigma)$ \cite{epaps}, and Theorem~\ref{thm:1} holds with equality for any matching pair $E_{g(F)}$ and $C_{g(F)}$ \footnote{This defines a whole new class of operational coherence measures. Note, however, that not all these measures will satisfy properties (C3) or (C4) \cite{Fan14}, as the corresponding entanglement measures are not all monotonic under selective LOCC operations or convex; this depends on the specific form of $g(F)$.}.


\noindent\textbf{\emph{Conclusions}}\textbf{.}---In this Letter we have established a rigorous and general framework for the interconversion between two quantum resources, coherence on one hand, and entanglement on
the other hand, via incoherent operations. Our framework can be interpreted in both ways: on one
hand, it demonstrates the formal potential of coherence for entanglement generation
(although not necessarily useful in practical applications, as cheaper schemes for entanglement creation might be available); on the other hand, it demonstrates the usefulness of entanglement to obtain and validate measures of coherence. Building on
this connection, we proposed in fact a family of coherence quantifiers
in terms of the maximal entanglement that can be generated
by incoherent operations (see Fig.~\ref{cofig}). The proposed coherence quantifiers
satisfy all the necessary criteria to be \textit{bona fide}
coherence monotones \cite{Baumgratz2014}. In particular, the relative entropy of coherence and the geometric measure of coherence have been (re)defined and interpreted operationally in terms of the maximum converted distillable and geometric entanglement, respectively.

Our framework bears some resemblance with, and may be regarded as the
general finite-dimensional counterpart to, the established (qualitative
and quantitative) equivalence between input nonclassicality, intended
as superposition of optical coherent states, and output entanglement
created by passive quantum optical elements such as  beam splitters
\cite{Kim02,Asboth05,Vogel2,Kil15}. The results presented here
should also be compared to the scheme for activating distillable entanglement
via premeasurement interactions \cite{Piani2011,Alex2011,Gharibian2011}
from quantum discord, a measure of nonclassical correlations going
beyond entanglement \cite{Modi2012,Streltsov2014}. In the latter approach, which has attracted a large amount of attention recently \cite{Modi2012,Piani2012,Nakano2013,Adesso2014,Mataloni2015}, measures of discord in a composite system are defined in terms of the minimum entanglement created with an ancillary system via fixed premeasurement interactions defined as in Eq.~(\ref{eq:CNOT}), where the minimization is over local unitaries on the system regulating the control bases before the interaction. By contrast, in this work the reference basis is fixed, and a maximization of the output entanglement over all incoherent operations returns a measure of coherence for the initial system. One might combine the two approaches in order to define a unified framework for interconversion among coherence, discord, and entanglement, whereby discord-type measures could be interpreted as measures of bipartite coherence suitably minimized over local product reference bases (see e.g. \cite{Frozen2014,GerTalk,NewBalance2015}). Exploring these connections further will be the subject of another work.

The theory of entanglement has been the cornerstone of major developments
in quantum information theory and has triggered the advancement of modern quantum technologies. The construction of a physically meaningful and mathematically rigorous
quantitative theory of coherence can improve our perception of genuine quantumness, and guide our understanding of
nascent fields such as quantum biology and nanoscale thermodynamics.
By uncovering a powerful operational connection between coherence
and entanglement, we believe
the present work delivers a substantial step in this direction.

\smallskip{}

\noindent\textbf{\emph{Acknowledgments}}\textbf{.}---We thank M.
Ahmadi, J. Asboth, R. Augusiak, T. Bromley, M. Cianciaruso, W. Vogel,
A. Winter, and especially M. Piani for fruitful discussions. A.S. acknowledges
financial support by the Alexander von Humboldt-Foundation, the John
Templeton Foundation, the EU (IP SIQS), the ERC AdG OSYRIS, and the
EU-Spanish Ministry (CHISTERA DIQIP). U.S., H.S.D., and M.N.B. acknowledge the Department of Atomic Energy, Government of India for research fellowship. G.A. acknowledges financial
support by the Foundational Questions Institute (Grant No.~FQXi-RFP3-1317) and
the ERC StG GQCOP (Grant Agreement No.~637352).


%


\cleardoublepage{}

\appendix*
\setcounter{equation}{0}
\pagenumbering{roman}
\setcounter{page}{1}

\section{Supplemental Material \\ Measuring Quantum Coherence with Entanglement}

\subsection{Proof of monotonicity (C3) in Theorem~3}

Here we prove that for any entanglement monotone $E$ the coherence
quantifier \emph{
\begin{equation}
C_{E}(\rho^{S})=\lim_{d_{A}\rightarrow\infty}\left\{ \sup_{\Lambda^{SA}}E^{S:A}\left(\Lambda^{SA}\left[\rho^{S}\otimes\ket{0}\bra{0}^{A}\right]\right)\right\} \label{eq:Ce-1}
\end{equation}
}does not increase on average under (selective) incoherent operations:
\begin{equation}
\sum_{i}p_{i}C_{E}(\sigma_{i}^{S})\leq C_{E}(\rho^{S})\label{eq:monotone}
\end{equation}
with probabilities $p_{i}=\mathrm{Tr}[K_{i}\rho^{S}K_{i}^{\dagger}]$,
quantum states $\sigma_{i}^{S}=K_{i}\rho^{S}K_{i}^{\dagger}/p_{i}$,
and incoherent Kraus operators $K_{i}$ acting on the system $S$.

Due to the definition of $C_{E}$, the amount of entanglement between
the system and ancilla cannot exceed $C_{E}$ for any incoherent operation
$\Lambda^{SA}$, i.e., $ $
\begin{equation}
E^{S:A}\left(\Lambda^{SA}\left[\rho^{S}\otimes\ket{0}\bra{0}^{A}\right]\right)\leq C_{E}\left(\rho^{S}\right).
\end{equation}
Note that this statement is also true if we introduce another particle
$B$ in an incoherent state $\ket{0}\bra{0}^{B}$. Then, for any tripartite
incoherent operation $\Lambda^{SAB}$ it holds:
\begin{equation}
E^{S:AB}\left(\Lambda^{SAB}\left[\rho^{S}\otimes\ket{0}\bra{0}^{A}\otimes\ket{0}\bra{0}^{B}\right]\right)\leq C_{E}\left(\rho^{S}\right).\label{eq:proof-3}
\end{equation}

We will now prove the claim by contradiction, showing that a violation
of Eq.~(\ref{eq:monotone}) also implies a violation of Eq.~(\ref{eq:proof-3}).
If Eq.~(\ref{eq:monotone}) is violated, then by definition of $C_{E}$
there exists a set of incoherent operations $\Lambda_{i}^{SA}$ such
that the following inequality is true for $d_{A}$ large enough:
\begin{equation}
\sum_{i}p_{i}E^{S:A}\left(\Lambda_{i}^{SA}\left[\sigma_{i}^{S}\otimes\ket{0}\bra{0}^{A}\right]\right)>C_{E}\left(\rho^{S}\right).\label{eq:proof}
\end{equation}
In the next step we introduce an additional particle $B$ and use
the general relation
\begin{equation}
E^{S:AB}\left(\sum_{i}p_{i}\rho_{i}^{SA}\otimes\ket{i}\bra{i}^{B}\right)\geq\sum_{i}p_{i}E^{S:A}\left(\rho_{i}^{SA}\right)
\end{equation}
which is valid for any entanglement monotone $E$. With this in mind,
the inequality (\ref{eq:proof}) implies
\begin{equation}
E^{S:AB}\left(\sum_{i}p_{i}\Lambda_{i}^{SA}\left[\sigma_{i}^{S}\otimes\ket{0}\bra{0}^{A}\right]\otimes\ket{i}\bra{i}^{B}\right)>C_{E}\left(\rho^{S}\right).
\end{equation}
Recall that the states $\sigma_{i}^{S}$ are obtained from the state
$\rho^{S}$ by the means of an incoherent operation, and thus we can
use the relation $p_{i}\sigma_{i}^{S}=K_{i}\rho^{S}K_{i}^{\dagger}$
with incoherent Kraus operators $K_{i}$. This leads us to the following
expression:
\begin{equation}
E^{S:AB}\left(\sum_{i}\Lambda_{i}^{SA}\left[K_{i}\rho^{S}K_{i}^{\dagger}\otimes\ket{0}\bra{0}^{A}\right]\otimes\ket{i}\bra{i}^{B}\right)>C_{E}\left(\rho^{S}\right).\label{eq:proof-1}
\end{equation}

It is now crucial to note that the state on the left-hand side of
the above expression can be regarded as arising from a tripartite
incoherent operation $\Lambda^{SAB}$ acting on the total state $\rho^{S}\otimes\ket{0}\bra{0}^{A}\otimes\ket{0}\bra{0}^{B}$:
\begin{align}
 & \Lambda^{SAB}\left[\rho^{S}\otimes\ket{0}\bra{0}^{A}\otimes\ket{0}\bra{0}^{B}\right]\nonumber \\
 & =\sum_{i}\Lambda_{i}^{SA}\left[K_{i}\rho^{S}K_{i}^{\dagger}\otimes\ket{0}\bra{0}^{A}\right]\otimes\ket{i}\bra{i}^{B}.\label{eq:proof-2}
\end{align}
This can be seen explicitly by introducing the Kraus operators $M_{ij}$
corresponding to the operation $\Lambda^{SAB}$:
\begin{equation}
M_{ij}^{SAB}=L_{ij}^{SA}\left(K_{i}^{S}\otimes\openone^{A}\right)\otimes U_{i}^{B}.
\end{equation}
Here, $L_{ij}$ are incoherent Kraus operators corresponding to the
incoherent operation $\Lambda_{i}^{SA}$:
\begin{equation}
\Lambda_{i}^{SA}\left[\rho^{SA}\right]=\sum_{j}L_{ij}\rho^{SA}L_{ij}^{\dagger}.
\end{equation}
The unitaries $U_{i}^{B}$ are incoherent and defined as
\begin{equation}
U_{i}^{B}=\sum_{j=0}^{d_{B}-1}\ket{\mathrm{mod}(i+j,d_{B})}\bra{j}^{B}.\label{eq:Ui}
\end{equation}
With these definitions we see that $M_{ij}$ are indeed incoherent
Kraus operators. Moreover, it can be verified by inspection that the
incoherent operation $\Lambda^{SAB}$ arising from these Kraus operators
also satisfies Eq.~(\ref{eq:proof-2}).

Finally, using Eq.~(\ref{eq:proof-2}) in Eq.~(\ref{eq:proof-1})
we arrive at the following inequality:
\begin{equation}
E^{S:AB}\left(\Lambda^{SAB}\left[\rho^{S}\otimes\ket{0}\bra{0}^{A}\otimes\ket{0}\bra{0}^{B}\right]\right)>C_{E}\left(\rho^{S}\right).
\end{equation}
This is the desired contradiction to Eq.~(\ref{eq:proof-3}), and
completes the proof of property (C3) for $C_{E}$, thus establishing
that $C_{E}$ is a coherence monotone for any entanglement monotone
$E$.

\subsection{Proof of convexity (C4) in Theorem~3}

Here we show that the quantifier of coherence $C_{E}$ given in Eq.~(\ref{eq:Ce-1})
is convex for any convex entanglement measure $E$:
\begin{equation}
C_{E}\left(\sum_{i}p_{i}\rho_{i}^{S}\right)\leq\sum_{i}p_{i}C_{E}\left(\rho_{i}^{S}\right)\label{eq:convexity}
\end{equation}
for any quantum states $\rho_{i}^{S}$ and probabilities $p_{i}$.
For this, note that by convexity of the entanglement quantifier $E$
it follows:
\begin{align}
 & E^{S:A}\left(\Lambda^{SA}\left[\sum_{i}p_{i}\rho_{i}^{S}\otimes\ket{0}\bra{0}^{A}\right]\right)\nonumber \\
 & \leq\sum_{i}p_{i}E^{S:A}\left(\Lambda^{SA}\left[\rho_{i}^{S}\otimes\ket{0}\bra{0}^{A}\right]\right).
\end{align}
Taking the supremum over all incoherent operations $\Lambda^{SA}$
together with the limit $d_{A}\rightarrow\infty$ on both sides of
this inequality we obtain the following result:
\begin{equation}
 C_{E}\left(\sum_{i}p_{i}\rho_{i}^{S}\right)\label{eq:convexity_proof} \leq\lim_{d_{A}\rightarrow\infty}\sup_{\Lambda^{SA}}\left\{ \sum_{i}p_{i}E^{S:A}\left(\Lambda^{SA}\left[\rho_{i}^{S}\otimes\ket{0}\bra{0}^{A}\right]\right)\right\}
\end{equation}
Finally, note that the right-hand side of this inequality cannot decrease
if the supremum over incoherent operations $\Lambda^{SA}$ and the
limit $d_{A}\rightarrow\infty$ are performed on each term of the
sum individually:
\begin{align}
 & \lim_{d_{A}\rightarrow\infty}\sup_{\Lambda^{SA}}\left\{ \sum_{i}p_{i}E^{S:A}\left(\Lambda^{SA}\left[\rho_{i}^{S}\otimes\ket{0}\bra{0}^{A}\right]\right)\right\} \nonumber \\
 & \leq\sum_{i}p_{i}\lim_{d_{A}\rightarrow\infty}\sup_{\Lambda^{SA}}E^{S:A}\left(\Lambda^{SA}\left[\rho_{i}^{S}\otimes\ket{0}\bra{0}^{A}\right]\right)\\
 & =\sum_{i}p_{i}C_{E}\left(\rho_{i}^{S}\right).\nonumber
\end{align}
Together with Eq.~(\ref{eq:convexity_proof}), this completes the
proof of convexity in Eq.~(\ref{eq:convexity}).

\subsection{Geometric entanglement and coherence}

We will now show that the bound provided in Theorem~1 of the main
text can be saturated for the geometric entanglement $E_{\mathrm{g}}$
and geometric coherence $C_{\mathrm{g}}$. The former is defined as
\begin{equation}
E_{\mathrm{g}}(\rho)=1-\max_{\chi\in{\cal S}}F(\rho,\chi),\label{eq:distance}
\end{equation}
where the maximum is taken over all separable states ${\cal S}$,
and
\begin{equation}
F(\rho,\varsigma)=\left(\mathrm{Tr}\sqrt{\sqrt{\rho}\varsigma\sqrt{\rho}}\right)^{2}
\end{equation}
is the fidelity. The geometric entanglement defined in Eq.~(\ref{eq:distance}) coincides with its expression obtained via a convex roof construction \cite{Streltsov2010}:
\begin{equation}
E_{\mathrm{g}}(\rho)=\min\sum_{i}p_{i}E_{\mathrm{g}}(\ket{\psi_{i}}),\label{eq:convex}
\end{equation}
where the minimum is taken over all decompositions of the state $\rho=\sum_{i}p_{i}\ket{\psi_{i}}\bra{\psi_{i}}$.
The latter expression (\ref{eq:convex}) is the original definition
of the geometric entanglement for mixed states \cite{Wei2003}, and
the equivalence of Eqs.~(\ref{eq:distance}) and (\ref{eq:convex})
was shown in \cite{Streltsov2010}.

The geometric coherence $C_{\mathrm{g}}$
can be defined similarly:
\begin{equation}
C_{\mathrm{g}}(\rho)=1-\max_{\sigma\in{\cal I}}F(\rho,\sigma),
\end{equation}
where the maximum is taken over all incoherent states $\sigma \in {\cal I}$.

Equipped with these tools, we are now in position to prove the saturation
of the bound in Theorem~1 for these measures of entanglement and coherence.
In particular, we will show the existence of an incoherent operation
$\Lambda^{SA}$ such that
\begin{equation}
E_{\mathrm{g}}^{S:A}\left(\Lambda^{SA}\left[\rho^{S}\otimes\ket{0}\bra{0}^{A}\right]\right)=C_{\mathrm{g}}\left(\rho^{S}\right)\label{eq:main}
\end{equation}
if the dimension of the ancilla is not smaller than the dimension
of the system, $d_{A}\geq d_{S}$. As we will further show,
the optimal incoherent operation achieving this equality is the generalized
CNOT operation $U$ given in Eq.~(\ref{eq:CNOT}) of the main text.

To prove Eq.~(\ref{eq:main}), we first recall that the final
state after the application of the generalized CNOT operation $U$
is
\begin{equation}
\rho_{\mathrm{mc}}^{SA}=U\left(\rho^{S}\otimes\ket{0}\bra{0}^{A}\right)U^{\dagger}=\sum_{i,j}\rho_{ij}\ket{i}\bra{j}^{S}\otimes\ket{i}\bra{j}^{A},\label{eq:mc}
\end{equation}
where $\rho_{ij}$ are matrix elements of $\rho^{S}=\sum_{i,j}\rho_{ij}\ket{i}\bra{j}^{S}$.
States of the form (\ref{eq:mc}) are known as maximally correlated
states. As shown in section \ref{sec:2} of this Supplemental Material,
there always exists a separable maximally correlated
state
\begin{equation}
\chi_{\mathrm{mc}}^{SA}=\sum_{i}q_{i}\ket{i}\bra{i}^{S}\otimes\ket{i}\bra{i}^{A}\label{eq:sigma}
\end{equation}
 which is a closest separable state to $\rho_{\mathrm{mc}}^{SA}$:
\begin{equation}
E_{\mathrm{g}}^{S:A}\left(\rho_{\mathrm{mc}}^{SA}\right)=1-F\left(\rho_{\mathrm{mc}}^{SA},\chi_{\mathrm{mc}}^{SA}\right).\label{eq:Eg-2}
\end{equation}

These results imply that the geometric coherence
of  $\rho^{S}$ is bounded above by the geometric entanglement
of $\rho_{\mathrm{mc}}^{SA}$:
\begin{equation}
C_{\mathrm{g}}\left(\rho^{S}\right)\leq E_{\mathrm{g}}^{S:A}\left(\rho_{\mathrm{mc}}^{SA}\right).
\end{equation}
This follows by using Eq.~(\ref{eq:Eg-2}) together with the
equality
\begin{equation}
F\left(\rho_{\mathrm{mc}}^{SA},\sigma_{\mathrm{mc}}^{SA}\right)=F\left(\rho^{S},\sigma^{S}\right),
\end{equation}
where $\sigma^{S}=\sum_{i}q_{i}\ket{i}\bra{i}^{S}$ is an incoherent
state with the same coefficients $q_{i}$ as in Eq.~(\ref{eq:sigma}).
On the other hand, Theorem~1 in the main text implies the inequality
\begin{equation}
C_{\mathrm{g}}\left(\rho^{S}\right)\geq E_{\mathrm{g}}^{S:A}\left(\rho_{\mathrm{mc}}^{SA}\right),
\end{equation}
and thus we arrive at the desired statement in Eq.~(\ref{eq:main}).

We further note that the arguments presented above can also be applied
to any distance-based quantifiers of entanglement $E_{D}$ and coherence
$C_{D}$ if the distance is contractive under quantum operations,
and for any maximally correlated state $\rho_{\mathrm{mc}}^{SA}$
there exists a separable maximally correlated state $\chi_{\mathrm{mc}}^{SA}\in{\cal S}$
(which may depend on $\rho_{\mathrm{mc}}^{SA}$) such that
\begin{equation}
E_{D}^{S:A}\left(\rho_{\mathrm{mc}}^{SA}\right)=D\left(\rho_{\mathrm{mc}}^{SA},\chi_{\mathrm{mc}}^{SA}\right).
\end{equation}
In particular, this is the case for the geometric entanglement and
coherence, where the distance is given by $D(\rho,\sigma)=1-F(\rho,\sigma)$.
As discussed in the main text, these results can be immediately extended to any distance
\begin{equation}
D(\rho,\sigma)=g[F(\rho,\sigma)]
\end{equation}
which is a nonincreasing function of fidelity.


\subsection{\label{sec:2}Geometric entanglement for maximally correlated states }

In this section we will show that for any maximally correlated state
\begin{equation}
\rho_{\mathrm{mc}}=\sum_{i,j}\rho_{ij}\ket{i}\bra{j}\otimes\ket{i}\bra{j}
\end{equation}
 there exists a separable maximally correlated state
\begin{equation}
\chi_{\mathrm{mc}}=\sum_{i}q_{i}\ket{i}\bra{i}\otimes\ket{i}\bra{i}
\end{equation}
 such that
\begin{equation}
E_{\mathrm{g}}(\rho_{\mathrm{mc}})=1-F(\rho_{\mathrm{mc}},\chi_{\mathrm{mc}}).\label{eq:Eg}
\end{equation}
This can be proven by using results from Refs. \cite{Horodecki2003,Streltsov2010}.
In particular, given a maximally correlated state $\rho_{\mathrm{mc}}$,
consider its arbitrary decomposition into pure states $\ket{\psi_{k}}$
with positive probabilities $p_{k}>0$ such that
\begin{equation}
\rho_{\mathrm{mc}}=\sum_{k}p_{k}\ket{\psi_{k}}\bra{\psi_{k}}.
\end{equation}
As is proven on page 6 in \cite{Horodecki2003}, all states $\ket{\psi_{k}}$
in such a decomposition must be linear combinations of product states
$\ket{i}\otimes\ket{i}$:
\begin{equation}
\ket{\psi_{k}}=\sum_{i}c_{i}^{k}\ket{i}\otimes\ket{i}\label{eq:psi}
\end{equation}
with complex coefficients $c_{i}^{k}$.
Consider now an optimal decomposition of the state $\rho_{\mathrm{mc}}$,
i.e., a decomposition which minimizes the average entanglement such
that
\begin{equation}
\sum_{k}p_{k}E_{\mathrm{g}}(\ket{\psi_{k}}\bra{\psi_{k}})=E_{\mathrm{g}}(\rho_{\mathrm{mc}}).
\end{equation}
We further define product states $\ket{\phi_{k}}\in{\cal S}$ to be
the closest product states to $\ket{\psi_{k}}$:
\begin{equation}
E_{\mathrm{g}}(\ket{\psi_{k}})=1-F(\ket{\psi_{k}},\ket{\phi_{k}}).\label{eq:Eg-1}
\end{equation}
Due to Eq.~(\ref{eq:psi}), all states $\ket{\phi_{k}}$ can be chosen
as $\ket{\phi_{k}}=\ket{l_{k}}\otimes\ket{l_{k}}$, where the number
$l_{k}$ corresponds to the coefficient $c_{l_{k}}^{k}$ of the state
$\ket{\psi_{k}}$ with the maximal absolute value: $|c_{l_{k}}^{k}|=\max_{i}|c_{i}^{k}|$.

Finally, consider the separable maximally correlated state
\begin{equation}
\chi_{\mathrm{mc}}=\sum_{k}q_{k}\ket{\phi_{k}}\bra{\phi_{k}}=\sum_{k}q_{k}\ket{l_{k}}\bra{l_{k}}\otimes\ket{l_{k}}\bra{l_{k}}
\end{equation}
with probabilities $q_{k}$ defined as
\begin{equation}
q_{k}=\frac{p_{k}\left[1-E_{\mathrm{g}}(\ket{\psi_{k}})\right]}{1-E_{\mathrm{g}}(\rho_{\mathrm{mc}})}.\label{eq:q}
\end{equation}
As we  now show, $\chi_{\mathrm{mc}}$ is the desired optimal
state, satisfying the equality (\ref{eq:Eg}). This can be seen by
first recalling that the geometric entanglement of the state $\rho_{\mathrm{mc}}$
is bounded above as
\begin{equation}
E_{\mathrm{g}}(\rho_{\mathrm{mc}})\leq1-F(\rho_{\mathrm{mc}},\chi_{\mathrm{mc}}),\label{eq:bound-1}
\end{equation}
since the state $\chi_{\mathrm{mc}}$ is separable. On the other
hand, the square root of the fidelity $\sqrt{F}$ satisfies the
strong concavity relation
\begin{equation}
\sqrt{F(\rho,\varsigma)}\geq\sum_{k}\sqrt{p_{k}q_{k}F(\ket{\psi_{k}},\ket{\phi_{k}})}
\end{equation}
for any two states $\rho=\sum_{k}p_{k}\ket{\psi_{k}}\bra{\psi_{k}}$
and $\varsigma=\sum_{k}q_{k}\ket{\phi_{k}}\bra{\phi_{k}}$; see Theorem~9.7 on page 414 in \cite{Nielsen10} (note that the fidelity
defined there is the square root of the fidelity used in our paper).  Applied to the maximally correlated states $\rho_{\mathrm{mc}}$
and $\chi_{\mathrm{mc}}$, and using  Eqs.~(\ref{eq:Eg-1})
and (\ref{eq:q}), this inequality becomes
\begin{equation}
\sqrt{F(\rho_{\mathrm{mc}},\chi_{\mathrm{mc}})}\geq\sqrt{1-E_{\mathrm{g}}(\rho_{\mathrm{mc}})},
\end{equation}
implying that the geometric entanglement of $\rho_{\mathrm{mc}}$
is bounded below as
\begin{equation}
E_{\mathrm{g}}(\rho_{\mathrm{mc}})\geq1-F(\rho_{\mathrm{mc}},\chi_{\mathrm{mc}}).
\end{equation}
Combining this result with Eq.~(\ref{eq:bound-1}) completes the
proof of Eq.~(\ref{eq:Eg}).

\subsection{Geometric coherence for arbitrary single-qubit states}
Earlier we have proven that the optimal incoherent operation which attains the maximization in Eq.~(\ref{eq:Ce-1}) is the generalized CNOT when $E$ is the geometric measure of entanglement $E_{\mathrm{g}}$.
If the system $S$ is a single qubit ($d_S=2$), the output state of system and ancilla after the CNOT is a two-qubit state.
The geometric entanglement $E_{\mathrm{g}}$ of any bipartite two-qubit state $\varsigma$ is computable in closed form and given by \cite{Wei2003,Streltsov2010}
\begin{align}
E_g^{S:A}(\varsigma) = \frac{1}{2}\left[1- \sqrt{1- \mathcal{C}(\varsigma)^2}\right]\,,
\end{align}
where $\mathcal{C}$ is the concurrence of $\varsigma$ \cite{Wootters97}.

Let $\rho^S = \sum_{i,j=0}^{1}\rho_{ij}\ket{i}\bra{j}^S$ be an arbitrary state of the single qubit $S$, written
with respect to a reference basis $\{\ket{i}\}$.
After applying the CNOT on the above state and the initially incoherent ancilla, we get a maximally correlated two-qubit state
\begin{align}
 \rho_{\mathrm{mc}}^{SA} = \mathrm{CNOT}[\rho^S\otimes\ket{0}\bra{0}^A] = \sum_{i,j=0}^{1}\rho_{ij}\ket{i}\bra{j}^{S}\otimes \ket{i}\bra{j}^{A}.
\end{align}
The concurrence of the above maximally correlated state is easily evaluated as \cite{Wootters97}
\begin{align}
 \mathcal{C}(\rho_{\mathrm{mc}}^{SA}) = 2|\rho_{01}|\,.
\end{align}
The geometric entanglement of $\rho_{\mathrm{mc}}^{SA}$ can be written then as
\begin{align}
 E_g^{S:A}(\rho_{\mathrm{mc}}^{SA}) &= \frac{1}{2}\left[1- \sqrt{1- \mathcal{C}(\rho_{\mathrm{mc}}^{SA})^2}\right]\nonumber\\
 &=\frac{1}{2}\left[1- \sqrt{1- 4|\rho_{01}|^2}\right].
\end{align}
Finally, the geometric coherence for an arbitrary single-qubit state $\rho$ is given by the expression reported in the main text:
\begin{align}
 C_g(\rho) = E_g^{S:A}(\rho_{\mathrm{mc}}^{SA}) =\frac{1}{2}\left[1- \sqrt{1- 4|\rho_{01}|^2}\right].
\end{align}

\end{document}